\pdfoutput=1
\documentclass[runningheads]{llncs}
\usepackage{amsmath}
\allowdisplaybreaks
\usepackage{amssymb}
\usepackage{bm}
\usepackage{graphicx}
\usepackage{booktabs}
\usepackage{cite}
\graphicspath{{../../eddr/figs/no_sips/}{../../eddr/figs/}{./figs/}}
\usepackage[colorlinks=true,allcolors=blue]{hyperref}

\usepackage{microtype}

\makeatletter
\let\corollary\@undefined \let\endcorollary\@undefined
\let\c@corollary\@undefined \let\thecorollary\@undefined
\makeatother
\spnewtheorem{corollary}{Corollary}[proposition]{\bfseries}{\itshape}
\renewcommand{\thecorollary}{\theproposition.\arabic{corollary}}

\begin{document}
\title{Recommendation Ranking Off-Policy Evaluation\texorpdfstring{\\}{ }
under Ranking-Dependent Examination\texorpdfstring{\\}{ }
via Examination-Relevance Decomposition}
\titlerunning{Recommendation Ranking OPE under Ranking-Dependent Examination}
    \author{Riki Okamura \and
            Toshiharu Sugawara}
  \authorrunning{R. Okamura and T. Sugawara}
  \institute{Department of Computer Science and Communications Engineering,
    Waseda University, Tokyo, Japan\\
    \email{r.okamura@isl.cs.waseda.ac.jp}, \email{sugawara@waseda.jp}}

\maketitle

\begin{abstract}
  \emph{Off-policy evaluation}, which estimates evaluation policy performance from logged data, is key for recommender ranking policies. However, logged clicks cannot distinguish unexamined items from examined non-clicks, causing bias in existing estimators when the assumed examination structures fail. We propose two estimators based on the decomposition of clicks into examination and relevance. First, the \emph{latent-examination independent inverse propensity score} (LE-IIPS) estimator corrects the IIPS bias using policy examination probability ratios. Second, the \emph{examination-decomposed doubly robust} (ED-DR) estimator extends LE-IIPS to a doubly robust framework. ED-DR is unbiased if the examination probabilities are correct regardless of relevance accuracy, or under ranking-independent examination, even if both model estimates are inaccurate. Experiments show that ED-DR achieves a lower MSE than existing methods with large sample sizes, especially when the examination depends on ranking. We also highlight its limitations under small samples or \emph{cascade user behavior} conditions.

\end{abstract}
\section{Introduction}\label{sec:introduction}
Ranking policies are widely used in recommender systems, such as those
of Amazon and Netflix. When introducing a new evaluation policy into
such a system, assessing its value through online experiments, such as
A/B tests, entails risks to user experience and operational
costs. \emph{Off-policy evaluation} (OPE) avoids these risks and costs
by estimating the value of an evaluation policy from logged data
collected by a deployed logging policy~\cite{strehl2010, dudik2014}.
\par

Logs from ranking policies record clicks at each position, but an
unclicked item does not reveal whether the user examined it and chose
not to click it or never examined it at all. Even interaction logs,
such as scrolling records, cannot confirm whether a displayed item was
examined. Despite this ambiguity, existing ranking OPE estimators,
such as \emph{independent inverse propensity score}
(IIPS)~\cite{iips}, \emph{reward-interaction inverse propensity
scoring} (reward-interaction IPS, RIPS)~\cite{rips}, \emph{cascade
doubly robust} (cascade DR)~\cite{cascadedr}, and \emph{adaptive IPS}
(AIPS)~\cite{aips}, describe user behavior solely by which items
affect the click at each position, without explicitly modeling
examination. For instance, IIPS assumes that the click at a position
depends only on the item presented there, an assumption that holds
under the \emph{position-based model} (PBM) where examination depends
only on position. However, when examination probabilities depend on
the broader ranking, such as when neighboring items draw attention
away, existing estimators become biased.
\par

In this study, we propose two estimators based on a click model that
decomposes clicks into examination and relevance~\cite{craswell2008,
  clickmodels}. First, the \emph{latent-examination IIPS} (LE-IIPS)
corrects the IIPS bias using policy examination probability
ratios. Second, the \emph{examination-decomposed doubly robust}
(ED-DR) estimator extends LE-IIPS to a \emph{doubly robust} (DR)
framework, maintaining unbiasedness when examination probabilities are
correctly estimated or under ranking-independent examination. Our
experiments using synthetic data show that, with large samples, ED-DR
achieves the lowest MSE among the compared estimators, especially when
examination depends on ranking. However, it underperforms existing
methods with small samples, and under \emph{cascade user behavior}, it exhibits
higher relative MSE than RIPS and cascade DR, clarifying both its
advantages and limitations.%
\footnote{\url{https://github.com/RickHub0115/ranking-policy-ope}}
\par
\section{Related Work}
Standard OPE estimators are the \emph{direct method}
(DM)~\cite{beygelzimer2009}, \emph{inverse propensity score}
(IPS)~\cite{horvitz1952}, and DR~\cite{dudik2011}. DM averages the
reward predictions under the evaluation policy and is biased when the
predictions are inaccurate. IPS reweights the observed rewards by the
ratio of the evaluation and logging policies'
probabilities~\cite{precup2000, strehl2010}; it is unbiased under
common support, but its variance grows as the two policies diverge. DR
combines the two to control both bias and variance~\cite{dudik2014}.
\par

In slate and ranking recommendations, where multiple items are
recommended simultaneously, the vast action space causes the variance
of IPS and DR to explode. For slate recommendation, where only a total
reward is observed, the \emph{pseudoinverse} estimator~\cite{pi}
decomposes the expected reward into per-slot contributions, a control
variate method reduces variance~\cite{vlassis}, and \emph{latent
IPS}~\cite{lips} compresses slates into low-dimensional
representations. For ranking recommendation, where rewards are
observed per position, IIPS~\cite{iips} assumes that a position's
reward depends only on the item presented there. RIPS~\cite{rips} and
cascade DR~\cite{cascadedr} relax this to a cascade assumption, where
rewards can also depend on higher-positioned items, while
AIPS~\cite{aips} selects behavior assumptions adaptively per context.
\par

Click models that decompose clicks into examination and relevance are
foundational in information
retrieval~\cite{clickmodels}. PBM~\cite{richardson2007,
  craswell2008} assumes examination depends solely on position. The
\emph{cascade model}~\cite{craswell2008} assumes a sequential
top-to-bottom examination that stops at the first relevant item. The
\emph{dynamic Bayesian network} (DBN) model~\cite{dbn} incorporates
post-click satisfaction, assuming that users leave only when
satisfied.
\par

Building on these click models, \emph{unbiased learning to rank}
(ULTR) learns ranking models from clicks weighted by inverse
examination probabilities~\cite{ultr}, which are often estimated via
regression EM~\cite{regem} or extended with DR-type
estimators~\cite{drultr}. Although ULTR leverages examination
estimates to optimize ranking algorithms, our study focuses on
off-policy evaluation, directly estimating the expected reward under a
target evaluation policy.
\par

In recommendation, observed feedback is biased because users select
which items to rate; \emph{unbiased recommender learning} corrects
this bias by reweighting feedback with observation
probabilities~\cite{rat, debiassurvey}. In particular, Saito et
al.~\cite{mnar} decomposed clicks into examination and relevance to
minimize the loss of unobserved relevance via inverse examination
weighting. However, these studies address missing user--item feedback
rather than evaluating policy performance.
\par

In summary, while ULTR and \emph{unbiased recommender learning} use
click decomposition to debias training, no prior OPE work explicitly
models examination and relevance as random variables. This study
bridges this gap by embedding this decomposition into OPE estimators.
\par
\section{Preliminaries}
We formalize OPE for rankings and review the IIPS estimator, which
underpins our proposed methods. The logged dataset $\mathcal{D} =
\{(x_i, \bm{a}_i, \bm{Y}_i)\}_{i=1}^{n}$ is generated by
\begin{equation}
  p(\mathcal{D}) = \prod_{i=1}^{n}
  p(x_i)
  \pi_0(\bm{a}_i \mid x_i)
  p(\bm{Y}_i \mid x_i, \bm{a}_i),
  \label{eq:dgp}
\end{equation}
where $n$ is the sample size and $x \sim p(x)$ denotes the context
(e.g., user profile). A ranking $\bm{a} = (\bm{a}(1),\dots,\bm{a}(K))
\in \Pi_K(\mathcal{A})$ of distinct items $\bm{a}(k) \in \mathcal{A}$
is chosen by the logging policy $\pi_0(\bm{a} \mid x)$. The vector
$\bm{Y} = (Y_1,\dots,Y_K)$ captures clicks, with $Y_k \in \{0,1\}$
indicating a click at position $k$. We define $q_k(x,\bm{a}) :=
\mathbb{E}[Y_k \mid x, \bm{a}]$.
\par

We define the value of an evaluation policy $\pi$ as
\begin{equation}
  V(\pi)
  := \mathbb{E}_{p(x)\,\pi(\bm{a} \mid x)\,p(\bm{Y} \mid x,\bm{a})}
  \!\left[\sum_{k=1}^{K} \alpha_k\, Y_k\right],
  \label{eq:policy-value}
\end{equation}
where $\alpha_k \ge 0$ is a position weight determined by the
provider; for example, $\alpha_k = 1$ gives the total number of
clicks, and $\alpha_k = 1/\log_2(k+1)$ gives DCG. We measure the
accuracy of an estimator $\hat V(\pi; \mathcal{D})$ of $V(\pi)$ by
$\mathrm{MSE}[\hat V] := \mathbb{E}_{p(\mathcal{D})}\big[(\hat
  V(\pi;\mathcal{D}) - V(\pi))^2\big] = \mathrm{Bias}(\hat V)^2 +
\mathbb{V}(\hat V)$, where $\mathrm{Bias}(\hat V) := \mathbb{E}[\hat
  V] - V(\pi)$ and $\mathbb{V}(\hat V) :=
\mathbb{E}\big[(\mathbb{E}[\hat V] - \hat V)^2\big]$.
\par

The IIPS estimator~\cite{iips}, a representative method for ranking OPE, assumes the following independence condition for clicks at position $k$:
\begin{equation}
  \forall x, \bm{a}, k,\quad
  \mathbb{E}[Y_k \mid x, \bm{a}]
  = \mathbb{E}[Y_k \mid x, \bm{a}(k)],
  \tag{A1}
  \label{as:a1}
\end{equation}
and is defined as
\begin{equation}
  \hat V_{\mathrm{IIPS}}(\pi; \mathcal{D})
  := \frac{1}{n}\sum_{i=1}^{n}\sum_{k=1}^{K}
  w_k^{\mathrm{IIPS}}(x_i, \bm{a}_i)
  \alpha_k Y_{i,k},
  \label{eq:iips}
\end{equation}
where $w_k^{\mathrm{IIPS}}(x, \bm{a}) := \pi(\bm{a}(k) \mid x,
k)/\pi_0(\bm{a}(k) \mid x, k)$ is the ratio of the marginal
probabilities $\pi(a \mid x, k) := \sum_{\bm{a}' \in
  \Pi_K(\mathcal{A})} \pi(\bm{a}' \mid x)\,\mathbb{I}\{\bm{a}'(k) =
a\}$ and $\pi_0(a \mid x, k)$ that item $a$ is presented at position
$k$. Under (A1) and the {\em marginal common support condition}
\begin{equation}
  \forall x, a, k,\quad
  \pi(a \mid x, k) > 0 \Rightarrow \pi_0(a \mid x, k) > 0,
  \tag{S1}
  \label{as:s1}
\end{equation}
IIPS is unbiased for $V(\pi)$. However, when the examination at
position $k$ depends on the full ranking, $\mathbb{E}[Y_k \mid x,
  \bm{a}]$ depends on items beyond $\bm{a}(k)$ alone; thus, (A1)
fails, and IIPS suffers from bias. To explicitly model this
dependence, Section~\ref{sec:clickmodel} introduces our click-model
framework.
\par
\section{Proposed Method}
\subsection{Click Models}\label{sec:clickmodel}
To resolve the ambiguity in $Y_k = 0$
(Section~\ref{sec:introduction}), we follow standard click
models~\cite{craswell2008, clickmodels} and formulate
\begin{equation}
  Y_k = O_k\, R_k,
  \label{eq:click-gen}
\end{equation}
where the latent variables $O_k \sim
\mathrm{Bern}\big(e_k(x,\bm{a})\big)$ and $R_k \sim
\mathrm{Bern}\big(r(x,\bm{a}(k))\big)$ denote examination and
relevance, respectively. Examination $e_k(x,\bm{a})$ may depend on
context and the full ranking, while relevance $r(x,\bm{a}(k))$ depends
only on context and item $\bm{a}(k)$. Thus, $Y_k = 0$ splits into
non-examination ($O_k = 0$) and non-relevance ($O_k = 1, R_k = 0$). We
assume:
\begin{equation}
  \forall x, \bm{a}, k,\quad
  O_k \perp R_k \mid x, \bm{a}
  \tag{A2}
  \label{as:a2}
\end{equation}
\begin{equation}
  \forall x, \bm{a}, k,\quad
  q_k(x,\bm{a}) = e_k(x,\bm{a})\, r(x,\bm{a}(k)).
  \tag{A3}
  \label{as:a3}
\end{equation}
Assumption (A3) follows from Eq.~\eqref{eq:click-gen} and (A2); see
Appendix~\ref{app:a3-derivation} for details.
\par

We classify click models by the variables governing $e_k$:
\begin{itemize}
\item[] (E1) PBM, $e_k(x,\bm{a}) = \theta_k$;
\item[] (E2) \emph{contextual PBM} (CPBM), $e_k(x,\bm{a}) = \theta_k(x)$; and
\item[] (E3) ranking-dependent examination,
\end{itemize}
where $e_k(x,\bm{a})$ depends on $\bm{a}$. Since $q_k(x,\bm{a}) =
e_k(x,\bm{a})\, r(x,\bm{a}(k))$ by (A3) and $r$ depends only on
$\bm{a}(k)$, (A1) holds if and only if $e_k$ does not depend on
$\bm{a}$: IIPS is unbiased under (E1) and (E2), but biased under
(E3). This work focuses on (E3), the most general case. In contrast,
the cascade model and DBN fall outside (E1)--(E3) because examination
at position $k$ depends on realized relevance at higher positions.
\par

Finally, the decomposition $q_k = e_k r$ in (A3) raises two
identifiability issues. First, scaling $(e_k, r) \to (c\, e_k, r/c)$
leaves $q_k$ unchanged; however, this ambiguity is harmless because
$V(\pi)$ depends on $(e_k, r)$ solely through $q_k$. Second,
identifying $e_k$ requires presenting the same item across multiple
positions with positive probability~\cite{ih2019, twotower}, which is
satisfied by stochastic policies (e.g., the \emph{Plackett--Luce}
model) but not deterministic ones. Lastly, $(e_k, r)$ is
unidentifiable when (A2) fails.
\subsection{Proposed Estimators}
Based on Assumption (A3), we propose two estimators. The first
estimator is LE-IIPS, which corrects the bias of IIPS under the click
model (E3) by the ratio of examination probabilities under the
evaluation and logging policies:
\begin{equation}
\hat V_{\mathrm{LE}}(\pi; \mathcal{D})
:= \frac{1}{n}\sum_{i=1}^{n}\sum_{k=1}^{K}
\hat w_k(x_i, \bm{a}_i)\,
\alpha_k\, Y_{i,k},
\label{eq:le-iips}
\end{equation}
where
\begin{equation}
\hat w_k(x,\bm{a})
:= \frac{\pi(\bm{a}(k) \mid x, k)}{\pi_0(\bm{a}(k) \mid x, k)}\,
\frac{\hat{\bar e}_k^{\pi}(x,\bm{a}(k))}{\hat e_k(x,\bm{a})},
\label{eq:what}
\end{equation}
$\hat e_k(x,\bm{a})$ estimates $e_k(x,\bm{a})$, 
and \mbox{$\hat{\bar e}_k^{\pi}(x, a) := \mathbb{E}_{\bm{a} \sim \pi(\bm{a} \mid x)}
  [\hat e_k(x, \bm{a}) \mid \bm{a}(k) = a]$} is the expectation of
$\hat e_k$ under evaluation policy $\pi$. When $\hat e_k = e_k$, we
write $\hat{\bar e}_k^{\pi}(x, a) = \mathbb{E}_{\bm{a} \sim \pi(\bm{a}
  \mid x)} [e_k(x, \bm{a}) \mid \bm{a}(k) = a] =: \bar e_k^{\pi}(x,
a)$. Under (E1) and (E2), $\hat e_k(x,\bm{a}) = \hat\theta_k(x)$ does
not depend on $\bm{a}$, so $\hat{\bar e}_k^{\pi}(x, \bm{a}(k)) = \hat
e_k(x, \bm{a})$ and LE-IIPS coincides with IIPS.
\par

However, because LE-IIPS ignores non-clicked observations ($Y_{i,k} =
0$), we define our second estimator, ED-DR, to utilize all
observations:
\begin{equation}
\hat V_{\mathrm{ED\text{-}DR}}(\pi; \mathcal{D})
:= \frac{1}{n}\sum_{i=1}^{n}\Big(
\mathbb{E}_{\pi(\bm{a} \mid x_i)}
\Big[\sum_{k=1}^{K} \alpha_k\, \hat q_k(x_i,\bm{a})\Big]
+ \sum_{k=1}^{K} \alpha_k\, \hat w_k
\big(Y_{i,k} - \hat q_k\big)
\Big),
\label{eq:eddr}
\end{equation}
where $\hat q_k(x,\bm{a}) := \hat e_k(x,\bm{a})\hat
r(x,\bm{a}(k))$, $\hat r(x,\bm{a}(k))$ estimates $r(x,\bm{a}(k))$, and
$\hat w_k$ is the weight in Eq.~\eqref{eq:what}. In the second term,
the arguments of $\hat w_k$ and $\hat q_k$ are $(x_i,\bm{a}_i)$.
\par

Conventional DR estimators directly regress expected clicks, whereas
ED-DR decomposes them into examination and relevance. This provides
two main advantages. First, because $\hat r(x,\bm{a}(k))$ is
position-independent, it pools data across all positions where an item
appears, improving the estimation accuracy for rare item-position
pairs. Second, ED-DR remains unbiased whenever $\hat e_k = e_k$,
regardless of $\hat r$'s accuracy (Proposition~\ref{prop:identity} in
Section~\ref{sec:theory}).
\par

Finally, because neither examination $O_k$ nor relevance $R_k$ is
observed, we estimate $(\hat e_k, \hat r)$ by regression
EM~\cite{regem}, with inputs $(x, k, \bm{a})$ and $(x, \bm{a}(k))$,
respectively, as in (A3). We approximate $\hat{\bar e}_k^{\pi}(x_i,
a)$ and the DM term of Eq.~\eqref{eq:eddr} by Monte Carlo averages
over $S$ rankings drawn from $\pi(\bm{a} \mid x_i)$, and use
cross-fitting~\cite{dml} to train $(\hat e_k, \hat r)$ and evaluate
the estimators.
\par
\subsection{Theoretical Analysis}\label{sec:theory}
We analyze the bias and variance of the proposed estimators and compare them with the MSEs of IIPS and ED-DR under (E3). Assuming cross-fitting, $(\hat e_k, \hat r)$ are learned on independent folds and treated as deterministic functions. Function arguments are omitted when they are clear from the context. 
All proofs are given in Appendix~\ref{app:proofs}.

\begin{proposition}[Unbiasedness of LE-IIPS]\label{prop:le}
Under (A2), (A3), and (S1), if $\hat e_k = e_k$, then
$\mathbb{E}_{p(\mathcal{D})}[\hat V_{\mathrm{LE}}(\pi; \mathcal{D})] = V(\pi)$.
\end{proposition}
The bias of LE-IIPS with respect to the error in $\hat e_k$ is bounded
by
\begin{equation}
  \big|\mathrm{Bias}\big(\hat V_{\mathrm{LE}}\big)\big|
  \le \sum_{k=1}^{K} \alpha_k\,
  \mathbb{E}\left[
    w_k^{\mathrm{IIPS}}
    \left| \frac{\hat{\bar e}_k^{\pi}}{\hat e_k}
    - \frac{\bar e_k^{\pi}}{e_k} \right|
    e_k\, r
    \right],
  \label{eq:le-bias-bound}
\end{equation}
which depends on the estimation error in the ratio $\bar
e_k^{\pi}/e_k$ rather than $\hat e_k$ itself.
\par

\begin{proposition}[Bias of ED-DR]\label{prop:identity}
Under (A2) and (A3), (S1), and the full common support condition (S2): $\forall x, \bm{a},\ \pi(\bm{a} \mid x) > 0 \Rightarrow \pi_0(\bm{a} \mid x) > 0$,
\begin{equation}
\mathrm{Bias}\big(\hat V_{\mathrm{ED\text{-}DR}}\big)
= \mathbb{E}_{p(x)\,\pi_0(\bm{a} \mid x)}\left[
\sum_{k=1}^{K} \alpha_k\,
\big(\hat w_k - w^{*}\big)\,\Delta_k
\right]
\label{eq:bias-identity}
\end{equation}
holds, where $w^{*}(x,\bm{a}) := \pi(\bm{a} \mid x)/\pi_0(\bm{a} \mid x)$ and $\Delta_k(x,\bm{a}) := q_k(x,\bm{a}) - \hat q_k(x,\bm{a})$. 
\end{proposition}
Thus, the bias equals the expected product of the weight error $\hat
w_k - w^{*}$ and the reward prediction error $\Delta_k$. We obtain the
following corollaries.
\par

\begin{corollary}[Unbiasedness under correct examination probabilities]
\label{cor:e-side}
Under the assumptions of Proposition~\ref{prop:identity}, if $\hat e_k
= e_k$, then $\mathrm{Bias}(\hat V_{\mathrm{ED\text{-}DR}}) = 0$.
\end{corollary}

\begin{corollary}[Unbiasedness under (E1) and (E2)]\label{cor:r-side}
Under the assumptions of Proposition 2, if the examination
probabilities satisfy (E1) or (E2), then $\mathrm{Bias}(\hat V_{\mathrm{ED\text{-}DR}}) = 0$.
\end{corollary}

These corollaries show that ED-DR is unbiased if the examination
probabilities are correctly estimated (Corollary~\ref{cor:e-side}) or
under ranking-independent examination (E1)/(E2) even with inaccurate
estimates (Corollary~\ref{cor:r-side}). Because weights $\hat w_k$
depend on $\hat e_k$, bias under (E3) persists when $\hat e_k \neq
e_k$ even if $\hat r = r$.  Thus,
relevance estimation errors affect only the variance, which we analyze next.
\par

\begin{proposition}[Variance analysis]
\label{prop:variance}
Under (S1), for any evaluation policy $\pi$,
\begin{align}
\mathbb{V}_{p(\mathcal{D})}\big(\hat V_{\mathrm{LE}}\big)
&= \frac{1}{n}\Big\{
\sigma_Y^2
+ \mathbb{E}_{p(x)}\Big[\mathbb{V}_{\pi_0(\bm{a} \mid x)}
\Big(\sum_{k=1}^{K} \alpha_k\, \hat w_k\, q_k\Big)\Big]
+ \sigma_x^2\Big\}
\label{eq:var-le}\\
\mathbb{V}_{p(\mathcal{D})}\big(\hat V_{\mathrm{ED\text{-}DR}}\big)
&= \frac{1}{n}\Big\{
\sigma_Y^2
+ \mathbb{E}_{p(x)}\Big[\mathbb{V}_{\pi_0(\bm{a} \mid x)}
\Big(\sum_{k=1}^{K} \alpha_k\, \hat w_k\, \Delta_k\Big)\Big]
+ \sigma_x^2\Big\}
\label{eq:var-eddr}
\end{align}
hold, where
\begin{equation*}
\sigma_Y^2 := \mathbb{E}_{p(x)\pi_0(\bm{a}|x)}
\big[\mathbb{V}_{p(\bm{Y}|x,\bm{a})}
\big(\textstyle\sum_{k} \alpha_k \hat w_k Y_k\big)\big],
\,
\sigma_x^2 := \mathbb{V}_{p(x)}
\big(\mathbb{E}_{\pi_0(\bm{a}|x)}
\big[\textstyle\sum_{k} \alpha_k \hat w_k q_k\big]\big).
\end{equation*}
\end{proposition}
Equations~\eqref{eq:var-le} and \eqref{eq:var-eddr} decompose the variance
into the click, logging-policy, and context components. Because the first
($\sigma_Y^2$) and third ($\sigma_x^2$) are identical for both
estimators, their variance difference depends solely on the second
term: weighted expected clicks $\sum_k \alpha_k \hat w_k q_k$ for
LE-IIPS versus weighted residuals $\sum_k \alpha_k \hat w_k \Delta_k$
for ED-DR. Thus, more accurate estimates $(\hat e_k, \hat r)$ that
minimize $\Delta_k$ yield a lower variance for ED-DR. 
\par

We next compare the variances of IIPS and ED-DR under (E1) and (E2),
where IIPS is already unbiased.

\begin{proposition}[Variance comparison]
\label{prop:var}
Under (S1), if examination probabilities satisfy (E1) or (E2), then for any evaluation policy $\pi$,
\begin{equation}
  \mathbb{V}_{p(\mathcal{D})}\big(\hat V_{\mathrm{ED\text{-}DR}}\big)
  - \mathbb{V}_{p(\mathcal{D})}\big(\hat V_{\mathrm{IIPS}}\big)
  = \frac{1}{n}\Big[
    \mathbb{V}_{p(x)\,\pi_0(\bm{a} \mid x)}(D - \hat D)
    - \mathbb{V}_{p(x)\,\pi_0(\bm{a} \mid x)}(D)\Big]
  \label{eq:var-comp}
\end{equation}
holds, where
$S := \sum_k \alpha_k\, w_k^{\mathrm{IIPS}}\, q_k$,
$\hat S := \sum_k \alpha_k\, w_k^{\mathrm{IIPS}}\, \hat q_k$,
$D := S - \mathbb{E}_{\pi_0(\bm{a} \mid x)}[S]$, and
$\hat D := \hat S - \mathbb{E}_{\pi_0(\bm{a} \mid x)}[\hat S]$.
\end{proposition}
Thus, ED-DR achieves lower variance than IIPS if and only if
$\mathbb{V}(D - \hat D) \le \mathbb{V}(D)$, which holds when $\hat D$
accurately predicts $D$. Conversely, ED-DR has higher variance if
$\hat D$ is negatively correlated with $D$.
\par

To compare the MSEs of IIPS and ED-DR under (E3), we express IIPS bias
as
\begin{equation}
  b_{\mathrm{IIPS}}
  := \mathbb{E}_{p(x)\,\pi(\bm{a} \mid x)}
  \Big[\sum_{k=1}^{K} \alpha_k\,
    \big(\bar e_k^{\pi_0}(x, \bm{a}(k)) - \bar e_k^{\pi}(x, \bm{a}(k))\big)\,
    r(x, \bm{a}(k))\Big],
  \label{eq:iips-bias}
\end{equation}
which measures examination probability differences between policies
and expands with policy divergence.

\begin{proposition}[MSE comparison between IIPS and ED-DR]
\label{prop:crossover}
Under the assumptions of Corollary~\ref{cor:e-side}, let
$\sigma_{\mathrm{EDDR}}^2/n
:= \mathbb{V}_{p(\mathcal{D})}\big(\hat V_{\mathrm{ED\text{-}DR}}\big)$ and
$\sigma_{\mathrm{IIPS}}^2/n
:= \mathbb{V}_{p(\mathcal{D})}\big(\hat V_{\mathrm{IIPS}}\big)$.
If $b_{\mathrm{IIPS}} \neq 0$ and
$\sigma_{\mathrm{EDDR}}^2 > \sigma_{\mathrm{IIPS}}^2$, then
\begin{equation}
\mathrm{MSE}_{\mathrm{ED\text{-}DR}} < \mathrm{MSE}_{\mathrm{IIPS}}
\;\iff\;
n > n^{*} :=
\frac{\sigma_{\mathrm{EDDR}}^2 - \sigma_{\mathrm{IIPS}}^2}{b_{\mathrm{IIPS}}^2}
\label{eq:crossover}
\end{equation}
holds.
\end{proposition}
Thus, ED-DR achieves lower MSE than IIPS for $n > n^{*}$, and for all $n$
when $\sigma_{\mathrm{EDDR}}^2 \le \sigma_{\mathrm{IIPS}}^2$. Because
$n^*$ shrinks as $b_{\mathrm{IIPS}}$ increases, ED-DR excels at
smaller sample sizes under greater policy divergence, as empirically
verified in Section~\ref{sec:experiment}.
\par
\section{Experiments}\label{sec:experiment}
We evaluate the proposed estimators using synthetic data. Beyond
confirming the theoretical behavior predicted in
Section~\ref{sec:theory}, we assess the robustness of the estimators
when assumptions such as (A1) and (A2) are not satisfied.

\subsection{Experimental Setup}\label{sec:exp-setup}
Following AIPS~\cite{aips}, contexts $x \sim \mathcal{N}(\bm{0},
I_5)$, with $m = 10$ items, ranking length $K = 5$, and $n = 4000$
logged samples, by default. The logging policy $\pi_0$ follows the
Plackett--Luce model:
\begin{equation}\notag
  \pi_0(\bm{a} \mid x)
  = \prod_{k=1}^{K}
  \frac{\exp\big(f(x, \bm{a}(k)) / \tau_0\big)}
       {\sum_{a \in \mathcal{A} \setminus \bm{a}(<k)} \exp\big(f(x, a) / \tau_0\big)},
\end{equation}
where $\bm{a}(<k)$ contains items placed above position $k$, $f(x, a)$
is a linear scoring function, and temperature $\tau_0$ (default $1$)
governs the policy stochasticity. The evaluation policy $\pi$ is
$\epsilon$-greedy on relevance ($\epsilon = 0.2$).
\par

The examination probabilities for (E1)--(E3) are specified as
\begin{itemize}
  \item[](E1) $e_k(x,\bm{a}) = 1/k^{\lambda}$;
  \item[](E2) $e_k(x,\bm{a}) = 1/k^{\lambda_{u(x)}}$;
  \item[](E3) $e_k(x,\bm{a}) = 1/k^{\lambda} \exp\big({-\beta
    \sum_{a \in \bm{a}(<k)}} \mathrm{attract}(x, a)\big)$,
\end{itemize}
where $\lambda=1$ (default), $u(x) \in \{1, 2\}$ via $x$'s first component with
$(\lambda_1, \lambda_2) = (\lambda/2, 2\lambda)$, and
$\beta=1$. Attractiveness $\mathrm{attract}(x, a) \in (0, 1)$
correlates with relevance $r(x, a)$~\cite{obp}. Clicks are generated
via $Y_k = O_k R_k$ with independent $O_k \sim \mathrm{Bern}(e_k)$ and
$R_k \sim \mathrm{Bern}(r)$, guaranteeing (A2) and (A3),
respectively. Unless stated otherwise, the experiments default to (E3)
and vary one hyperparameter at a time.
\par

We evaluate ED-DR and LE-IIPS (including their oracle
variants with true $e_k$ and $r$) against IIPS~\cite{iips},
RIPS~\cite{rips}, cascade DR~\cite{cascadedr},
AIPS~\cite{aips}, DM, and
DR-IIPS, a DR~\cite{dudik2011} extension of IIPS. The regression EM uses logistic
models with $S = 100$ Monte Carlo samples and intervention harvesting
initialization~\cite{ih2019}. The results are robust to these
implementation choices, including $S$, initialization, and default
position weights $\alpha_k = 1$ (figures omitted). The performance
is evaluated using relative MSE, $\mathbb{E}_{p(\mathcal{D})}
\big[\big(\hat V - V(\pi)\big)^2 / V(\pi)^2\big]$, alongside its
squared bias and variance, averaged over independent random seeds.
\par

\begin{figure}
\centering
\includegraphics[width=\columnwidth]{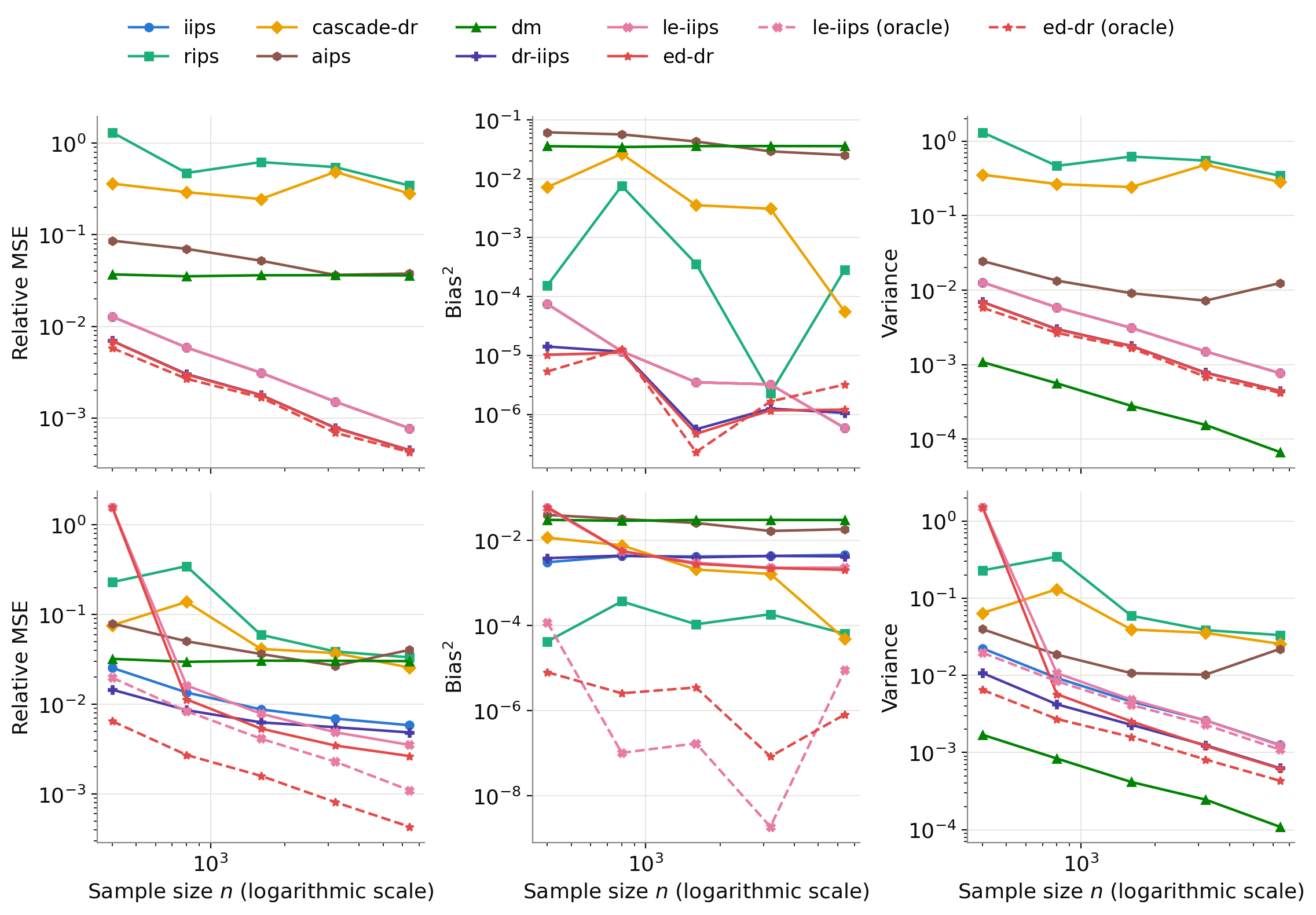}
\caption{Estimation error.
Top row: (E1) PBM; bottom row: (E3) ranking independent.}
\label{fig:data-size}
\end{figure}

\subsection{Results}\label{sec:exp-results}
Except for the oracle estimators, ED-DR consistently achieves
the lowest relative MSE. Performance degrades only when the underlying
assumptions are violated, with such degradation appearing solely as
increased bias rather than variance.
\par

\subsubsection{Effect of the sample size.}\label{sec:exp-n}
We varied $n \in \{400, 800, 1600, 3200, 6400\}$ under (E1)
(Fig.~\ref{fig:data-size}, top) and (E3) (bottom). Because (A2) and
(A3) hold, ED-DR and LE-IIPS are unbiased under both
settings, with variance decreasing as $1/n$. Under (E1), (A1) holds
and IIPS is unbiased, yielding comparable performance across
estimators; ED-DR matches IIPS despite using estimated
parameters $(\hat e_k, \hat r)$. Under (E3), (A1) fails, causing a
constant bias $b_{\mathrm{IIPS}}$ in
Eq.~\eqref{eq:iips-bias}. Consequently, ED-DR outperforms
IIPS for large $n$, whereas IIPS dominates at small $n$
owing to its lower variance, confirming the crossover $n^*$ from
Proposition~\ref{prop:crossover}.
\par

\begin{figure}
\centering
\includegraphics[width=\columnwidth]{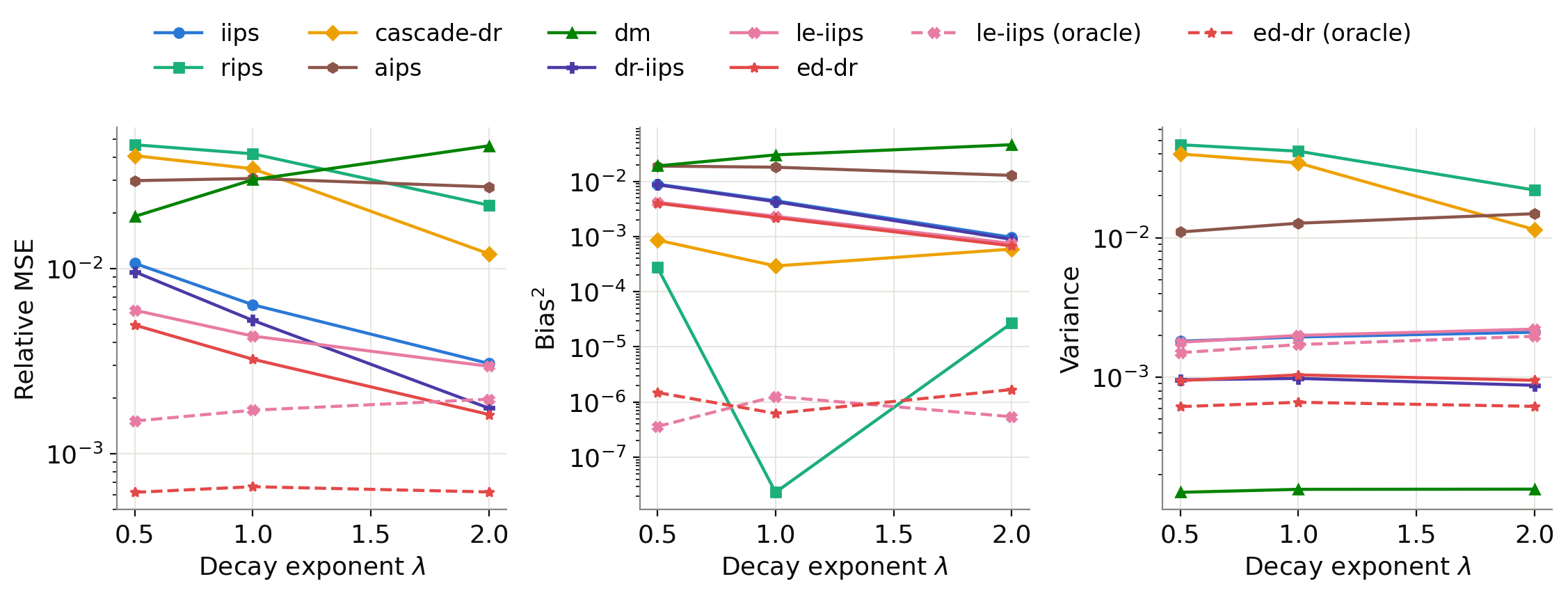}
\caption{Estimation error against the decay exponent $\lambda$.}
\label{fig:decay}
\end{figure}

\subsubsection{Effect of decay exponent $\lambda$.}\label{sec:exp-decay}
We varied $\lambda \in \{0.5, 1.0, 2.0\}$ under (E3)
(Fig.~\ref{fig:decay}). A larger $\lambda$ reduces examination at
lower positions and concentrates clicks at top ranks, decreasing
estimation error across most estimators. ED-DR consistently
achieves the lowest relative MSE, maintaining a slight edge over
DR-IIPS at $\lambda = 2$.
\par

\begin{figure}
\centering
\includegraphics[width=\columnwidth]{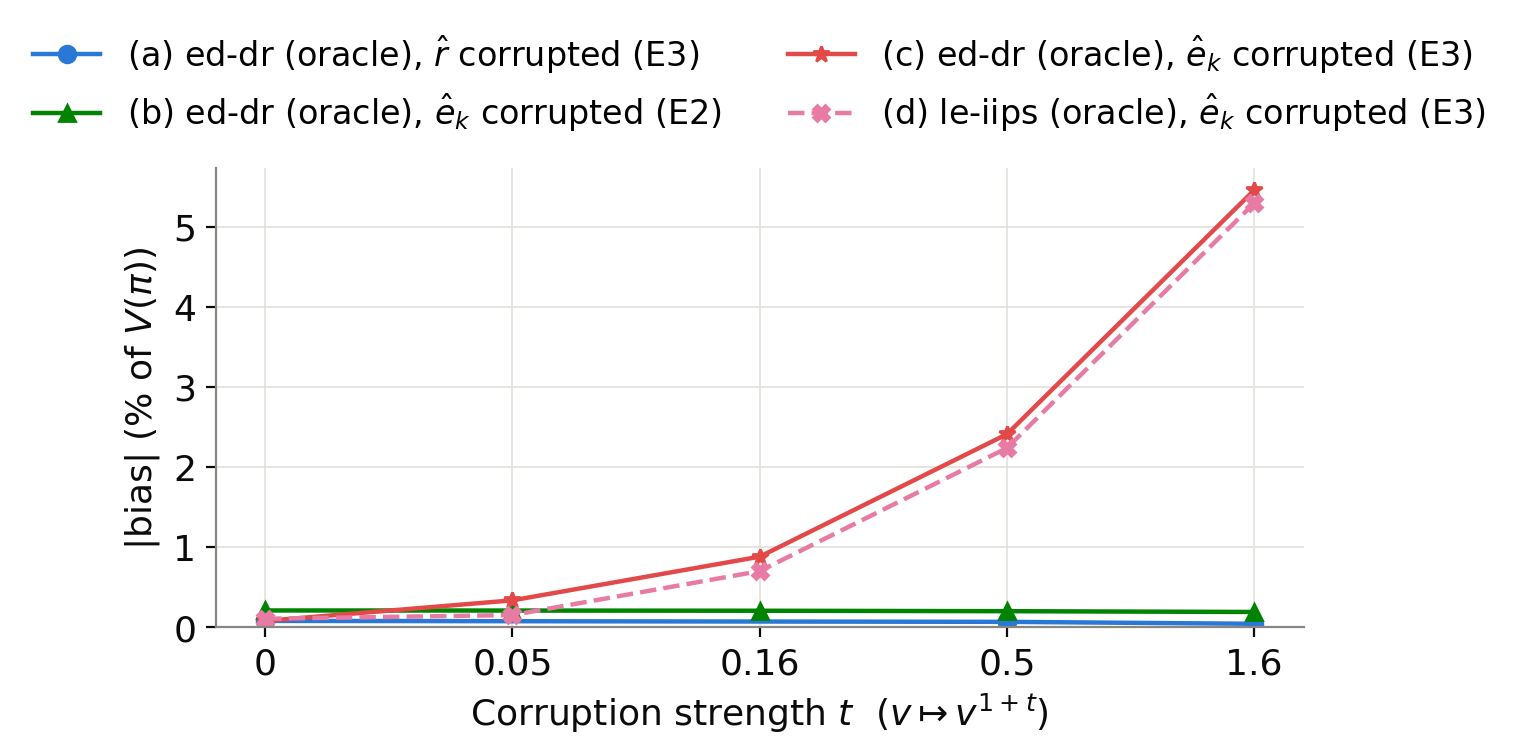}
\caption{$|\mathrm{Bias}|$ of ED-DR (oracle) against the size $t$ of the
error injected into one estimate. (a) $\hat r$ corrupted under (E3);
(b) $\hat e_k$ corrupted under (E2); (c) $\hat e_k$ corrupted under (E3),
with (d) LE-IIPS (oracle) given the same $\hat e_k$ for reference.}
\label{fig:injection}
\end{figure}

\subsubsection{Effect of errors in the examination and relevance estimates.}
\label{sec:exp-injection}
We empirically verified that ED-DR bias depends exclusively on
examination estimation errors. Starting from ED-DR (oracle), we
corrupt either $e_k$ or $r$ as $v \mapsto v^{1+t}$ ($t \ge 0$) and
plot $\vert\mathrm{Bias}\vert$ against $t$
(Fig.~\ref{fig:injection}). We consider three cases: (a) under (E3),
corrupting only $\hat r$; (b) under (E2), corrupting only $\hat e_k$
while maintaining ranking independence; and (c) under (E3), corrupting
only $\hat e_k$ alongside LE-IIPS (oracle) given the same
corrupted $\hat e_k$. Bias stays zero for all $t$ in (a) and (b),
growing with $t$ only in (c) where ED-DR matches
LE-IIPS. This confirms Corollaries~\ref{cor:e-side}
and~\ref{cor:r-side}: relevance errors never induce bias, examination
errors introduce no bias under (E2), and the bias under (E3) is
determined solely by $\hat e_k$.

\begin{table}[t]
\centering
\caption{Relative MSE of ED-DR ($\times 10^{-3}$) under misspecified
examination structures. Rows: true examination structure; columns:
examination structure assumed by EM.}
\label{tab:misspec}
\tabcolsep=5pt
\begin{tabular}{lccc}
\toprule
True $\backslash$ Assumed & (E1) & (E2) & (E3) \\
\midrule
(E1) & $0.73$ & $0.73$ & $0.84$ \\
(E2) & $0.65$ & $0.61$ & $0.72$ \\
(E3) & $5.24$ & $5.25$ & $3.24$ \\
\bottomrule
\end{tabular}
\end{table}

\subsubsection{Effect of misspecified examination probabilities.}
\label{sec:exp-misspec}
We independently varied the true and EM-assumed examination structures
over (E1)--(E3) (Table~\ref{tab:misspec}). Misspecification barely
affects the relative MSE under true (E1) and has a minor effect under
(E2), but substantially increases the error under (E3). Thus, the
misspecification loss increases with the complexity of the true
examination model. Figure~\ref{fig:bound} shows the evaluation of the
LE-IIPS bias bound from Eq.~\eqref{eq:le-bias-bound} across all
combinations in Table~\ref{tab:misspec}, plotting the theoretical
bound against the actual bias $\vert\mathrm{Bias}(\hat
V_{\mathrm{LE}})\vert$. Most points fall below $y = x$, empirically
validating that the misspecification bias remains within the
theoretical bound.
\par

\begin{figure}[t]
\centering
\includegraphics[width=0.5\columnwidth]{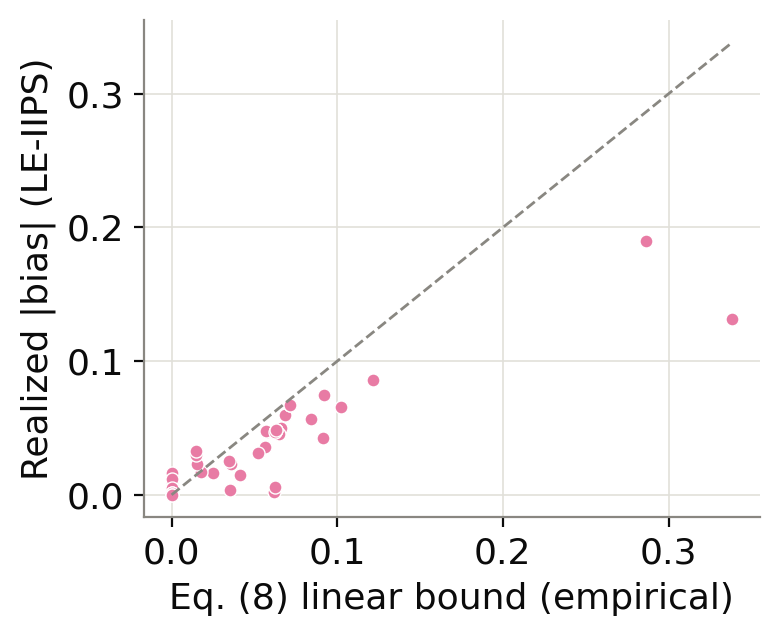}
\caption{Actual bias $|\mathrm{Bias}(\hat V_{\mathrm{LE}})|$ of LE-IIPS
against the bound of Eq.~\eqref{eq:le-bias-bound}.}
\label{fig:bound}
\end{figure}

\begin{figure}[t]
\centering
\includegraphics[width=\columnwidth]{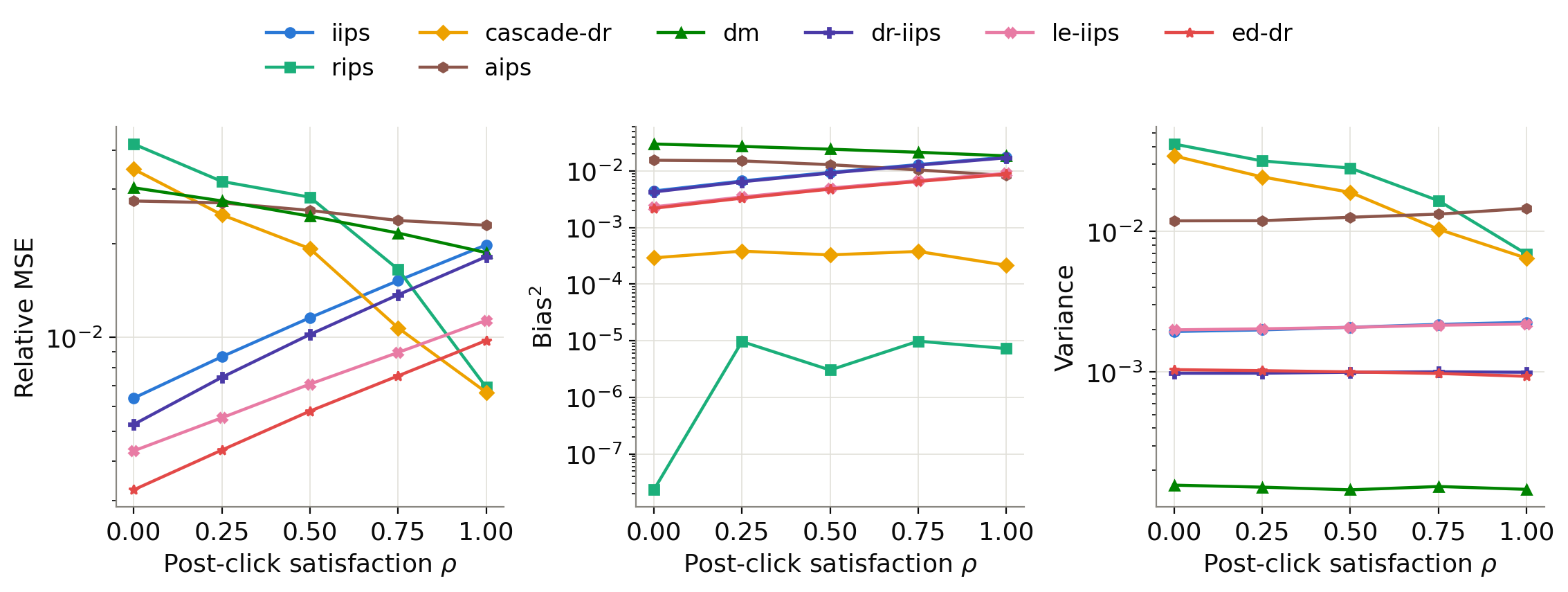}
\caption{Estimation error against the probability $\rho$ under DBN.}
\label{fig:a2}
\end{figure}

\subsubsection{Under DBN.}\label{sec:lose-a2}
We evaluated estimator performance under the DBN model
(Fig.~\ref{fig:a2}), where satisfied users leave with probability
$\rho$. A higher $\rho$ increases $O_k$'s dependence on
higher-position relevance $R_j$ ($j < k$), deviating further from
(A2). As $\rho$ grows, ED-DR's bias increases while variance stays
constant, confirming that violations of (A2) introduce bias rather
than variance. However, the bias of ED-DR increases more slowly than
that of IIPS, keeping ED-DR superior unless $\rho \to 1$. At $\rho =
1$, clicks occur only at top ranks, preventing weight accumulation at
lower positions for RIPS and cascade DR; their variance drops
significantly, outperforming ED-DR in relative MSE.
\par

\begin{figure}[t]
\centering
\includegraphics[width=\columnwidth]{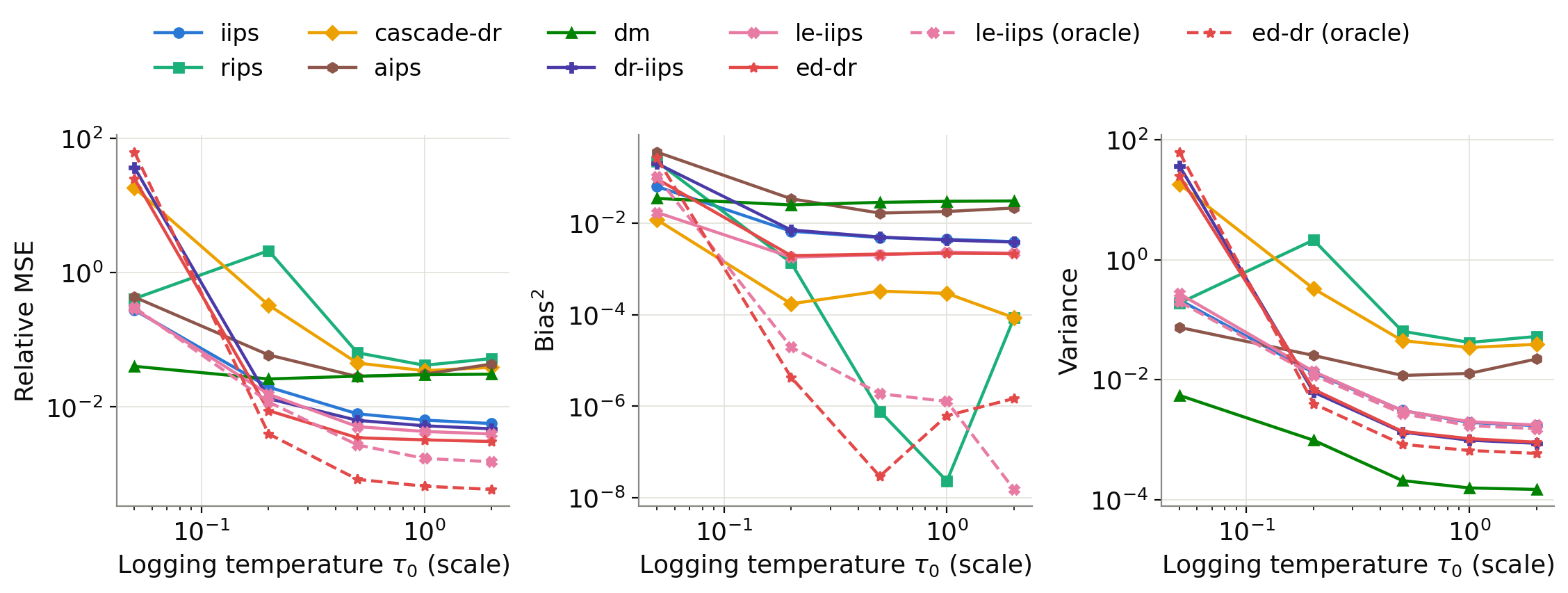}
\caption{Estimation error against the temperature $\tau_0$ of the logging
policy.}
\label{fig:tau0}
\end{figure}

\subsubsection{Effect of logging policy temperature $\tau_0$.}
\label{sec:lose-deterministic}
We varied the logging temperature $\tau_0$
(Fig.~\ref{fig:tau0}). Lower $\tau_0$ makes the logging policy
deterministic, rendering $e_k$ and $r$ unidentifiable
(Section~\ref{sec:clickmodel}) and causing importance weights to
diverge. At small $\tau_0$, all weight-based estimators suffer from
exploding variance, whereas DM remains stable but exhibits persistent
bias. For larger $\tau_0$, ED-DR consistently achieves the lowest
relative MSE. Here, the bias of ED-DR exceeds its variance; because
oracle ED-DR lacks this bias, it stems entirely from EM estimation.
\par
\section{Conclusion}
We addressed the ambiguity in ranking OPE, where logged clicks
confound unexamined items with examined non-clicks. By decomposing
clicks into examination and relevance, we proposed LE-IIPS, which
corrects IIPS bias, and its doubly robust extension, ED-DR. We
established that ED-DR is unbiased if examination probabilities are
accurately estimated or under ranking-independent examination
(Proposition~\ref{prop:identity}), and that it achieves a lower MSE
than IIPS for sample sizes exceeding $n^*$ under ranking-dependent
examination (Proposition~\ref{prop:crossover}). Experiments
demonstrated that ED-DR consistently achieves the lowest MSE for large
samples, falling behind only under small sample sizes, nearly
deterministic policies, or strong user departure dynamics in DBN.
\par

Future work includes dynamically selecting or posterior-averaging the
assumed examination structure per context (as in AIPS), and extending
the click decomposition to sequential examination to accommodate
cascade models and DBN.
\par
\appendix
\section{Appendix --- Omitted Proofs}
\label{app:proofs}

\subsection{Derivation of (A3)}
\label{app:a3-derivation}
\begin{align*}
q_k(x,\bm{a})
&= \mathbb{E}[O_k R_k \mid x, \bm{a}]
&& {\scriptstyle (\because\ \text{Eq.~\eqref{eq:click-gen}})}\\
&= \mathbb{E}[O_k \mid x, \bm{a}]\,
   \mathbb{E}[R_k \mid x, \bm{a}]
&& {\scriptstyle (\because\ \text{A2})}\\
&= e_k(x,\bm{a})\, r(x,\bm{a}(k))
&& {\scriptstyle (\because\ O_k,R_k\sim\mathrm{Bern})}
\end{align*}

\subsection{Auxiliary Lemmas}
The following two lemmas are used repeatedly in the proofs below.

\begin{lemma}
\label{lem:transport}
Under (S1), for any $x$, $k$, and any integrable function $f(x, a)$,
\begin{equation*}
\mathbb{E}_{\pi_0(\bm{a} \mid x)}
\big[w_k^{\mathrm{IIPS}}\, f(x, \bm{a}(k))\big]
= \mathbb{E}_{\pi(\bm{a} \mid x)}
\big[f(x, \bm{a}(k))\big]
\end{equation*}
holds.
\end{lemma}

\begin{proof}
\begin{align*}
&\mathbb{E}_{\pi_0(\bm{a} \mid x)}
\big[w_k^{\mathrm{IIPS}}\, f(x, \bm{a}(k))\big] \\
&= \sum_{\bm{a}} \pi_0(\bm{a} \mid x)\,
\tfrac{\pi(\bm{a}(k) \mid x, k)}{\pi_0(\bm{a}(k) \mid x, k)}\,
f(x, \bm{a}(k)) \\
&= \sum_{\bm{a}} \pi_0(\bm{a} \mid x)
\sum_{a \in \mathcal{A}}
\tfrac{\pi(a \mid x, k)}{\pi_0(a \mid x, k)}\,
\mathbb{I}\{\bm{a}(k) = a\}\, f(x, a) \\
&= \sum_{a \in \mathcal{A}}
\sum_{\bm{a}} \pi_0(\bm{a} \mid x)\,
\mathbb{I}\{\bm{a}(k) = a\}\,
\tfrac{\pi(a \mid x, k)}{\pi_0(a \mid x, k)}\, f(x, a) \\
&= \sum_{a \in \mathcal{A}} \pi_0(a \mid x, k)\,
\tfrac{\pi(a \mid x, k)}{\pi_0(a \mid x, k)}\, f(x, a) \\
&= \sum_{a \in \mathcal{A}} \sum_{\bm{a}}
\pi(\bm{a} \mid x)\, \mathbb{I}\{\bm{a}(k) = a\}\, f(x, a) \\
&= \sum_{\bm{a}} \pi(\bm{a} \mid x)\, f(x, \bm{a}(k))
= \mathbb{E}_{\pi(\bm{a} \mid x)}\big[f(x, \bm{a}(k))\big]
\end{align*}
\end{proof}

\begin{lemma}
\label{lem:tower}
For any $x$, $k$, and any integrable function $f(x, a)$,
\begin{equation*}
\mathbb{E}_{\pi(\bm{a} \mid x)}
\big[\bar e_k^{\pi}(x, \bm{a}(k))\, f(x, \bm{a}(k))\big]
= \mathbb{E}_{\pi(\bm{a} \mid x)}
\big[e_k(x, \bm{a})\, f(x, \bm{a}(k))\big]
\end{equation*}
holds.
\end{lemma}

\begin{proof}
\begin{align*}
\text{(LHS)}
&= \sum_{\bm{a}} \pi(\bm{a} \mid x)\,
\bar e_k^{\pi}(x, \bm{a}(k))\, f(x, \bm{a}(k)) \\
&= \sum_{a \in \mathcal{A}}
\sum_{\bm{a}} \pi(\bm{a} \mid x)\,
\mathbb{I}\{\bm{a}(k) = a\}\,
\bar e_k^{\pi}(x, a)\, f(x, a) \\
&= \sum_{a \in \mathcal{A}} \pi(a \mid x, k)\,
\bar e_k^{\pi}(x, a)\, f(x, a) \\
\text{(RHS)}
&= \sum_{\bm{a}} \pi(\bm{a} \mid x)\,
e_k(x, \bm{a})\, f(x, \bm{a}(k)) \\
&= \sum_{a \in \mathcal{A}}
\sum_{\bm{a}} \pi(\bm{a} \mid x)\,
e_k(x, \bm{a})\,
\mathbb{I}\{\bm{a}(k) = a\}\, f(x, a) \\
&= \sum_{a \in \mathcal{A}} \pi(a \mid x, k)\,
\bar e_k^{\pi}(x, a)\, f(x, a)
\;\; {\scriptstyle(\because\ \text{definition of } \bar e_k^{\pi})}
\end{align*}
Hence (LHS) $=$ (RHS).
\end{proof}

\subsection{Proof of Proposition~\ref{prop:le} and Eq.~\eqref{eq:le-bias-bound}}

\begin{proof}
\begin{align*}
&\mathbb{E}_{p(\mathcal{D})}\big[\hat V_{\mathrm{LE}}(\pi; \mathcal{D})\big] \\
&= \mathbb{E}_{p(\mathcal{D})}
\Big[\sum_{k} \alpha_k\,
\tfrac{\pi(\bm{a}(k) \mid x, k)}{\pi_0(\bm{a}(k) \mid x, k)}\,
\tfrac{\bar e_k^{\pi}(x,\, \bm{a}(k))}{e_k(x,\, \bm{a})}\, Y_k\Big]
\;\; {\scriptstyle(\because\ \text{i.i.d.})} \\
&= \mathbb{E}_{p(x)\,\pi_0(\bm{a} \mid x)}
\Big[\sum_{k} \alpha_k\, w_k^{\mathrm{IIPS}}\,
\tfrac{\bar e_k^{\pi}}{e_k}\, q_k(x, \bm{a})\Big] \\
&= \mathbb{E}_{p(x)\,\pi_0(\bm{a} \mid x)}
\Big[\sum_{k} \alpha_k\, w_k^{\mathrm{IIPS}}\,
\bar e_k^{\pi}\, r(x, \bm{a}(k))\Big]
\;\; {\scriptstyle(\because\ \text{(A3)})} \\
&= \mathbb{E}_{p(x)\,\pi(\bm{a} \mid x)}
\Big[\sum_{k} \alpha_k\,
\bar e_k^{\pi}\, r\Big]
\;\; {\scriptstyle(\because\ \text{Lemma~\ref{lem:transport}})} \\
&= \mathbb{E}_{p(x)\,\pi(\bm{a} \mid x)}
\Big[\sum_{k} \alpha_k\, e_k\, r\Big]
\;\; {\scriptstyle(\because\ \text{Lemma~\ref{lem:tower}})} \\
&= \mathbb{E}_{p(x)\,\pi(\bm{a} \mid x)}
\Big[\sum_{k} \alpha_k\, q_k\Big]
= V(\pi)
\;\; {\scriptstyle(\because\ \text{(A3)})}
\end{align*}
\end{proof}

\medskip\noindent\textit{Proof of Eq.~\eqref{eq:le-bias-bound}.}
\begin{align*}
\mathrm{Bias}\big(\hat V_{\mathrm{LE}}\big)
= \mathbb{E}_{p(\mathcal{D})}\big[\hat V_{\mathrm{LE}}(\pi; \mathcal{D})\big] - V(\pi)
\end{align*}
\begin{align*}
\text{(first term)}
&= \mathbb{E}_{p(x)\,\pi_0(\bm{a} \mid x)}
\Big[\sum_{k} \alpha_k\, w_k^{\mathrm{IIPS}}\,
\tfrac{\hat{\bar e}_k^{\pi}}{\hat e_k}\, e_k\, r\Big] \\
\text{(second term)}
&= \mathbb{E}_{p(x)\,\pi(\bm{a} \mid x)}
\Big[\sum_{k} \alpha_k\, e_k\, r\Big]
\;\; {\scriptstyle(\because\ \text{(A3)})} \\
&= \mathbb{E}_{p(x)\,\pi(\bm{a} \mid x)}
\Big[\sum_{k} \alpha_k\, \bar e_k^{\pi}\, r\Big]
\;\; {\scriptstyle(\because\ \text{Lemma~\ref{lem:tower}})} \\
&= \mathbb{E}_{p(x)\,\pi_0(\bm{a} \mid x)}
\Big[\sum_{k} \alpha_k\, w_k^{\mathrm{IIPS}}\,
\bar e_k^{\pi}\, r\Big]
\;\; {\scriptstyle(\because\ \text{Lemma~\ref{lem:transport}})} \\
&= \mathbb{E}_{p(x)\,\pi_0(\bm{a} \mid x)}
\Big[\sum_{k} \alpha_k\, w_k^{\mathrm{IIPS}}\,
\tfrac{\bar e_k^{\pi}}{e_k}\, e_k\, r\Big]
\end{align*}
\begin{align*}
\therefore\ \big|\mathrm{Bias}\big(\hat V_{\mathrm{LE}}\big)\big|
&= \Big|\mathbb{E}
\Big[\sum_{k} \alpha_k\, w_k^{\mathrm{IIPS}}
\Big(\tfrac{\hat{\bar e}_k^{\pi}}{\hat e_k}
- \tfrac{\bar e_k^{\pi}}{e_k}\Big)\, e_k\, r\Big]\Big| \\
&\le \sum_{k} \alpha_k\,
\mathbb{E}\Big[w_k^{\mathrm{IIPS}}
\Big|\tfrac{\hat{\bar e}_k^{\pi}}{\hat e_k}
- \tfrac{\bar e_k^{\pi}}{e_k}\Big|\, e_k\, r\Big]
\end{align*}
\qed

\subsection{Proof of Proposition~\ref{prop:identity} and Corollaries~\ref{cor:e-side} and~\ref{cor:r-side}}

\begin{proof}
\begin{align*}
&\mathbb{E}_{p(\mathcal{D})}\big[\hat V_{\mathrm{ED\text{-}DR}}(\pi; \mathcal{D})\big] \\
&= \mathbb{E}_{p(\mathcal{D})}
\Big[\mathbb{E}_{\bm{a}' \sim \pi(\bm{a}' \mid x)}
\Big[\sum_{k} \alpha_k\, \hat q_k(x, \bm{a}')\Big] \\
&\quad
+ \sum_{k} \alpha_k\, \hat w_k(x, \bm{a})\,
\big(Y_k - \hat q_k(x, \bm{a})\big)\Big]
\;\; {\scriptstyle(\because\ \text{i.i.d.})} \\
&= \mathbb{E}_{p(x)\,\pi(\bm{a} \mid x)}
\Big[\sum_{k} \alpha_k\,
\big(q_k - (q_k - \hat q_k)\big)\Big] \\
&\quad + \mathbb{E}_{p(x)\,\pi_0(\bm{a} \mid x)}
\Big[\sum_{k} \alpha_k\, \hat w_k\,
\big(q_k - \hat q_k\big)\Big] \\
&= V(\pi)
- \mathbb{E}_{p(x)\,\pi(\bm{a} \mid x)}
\Big[\sum_{k} \alpha_k\, \Delta_k\Big] \\
&\quad + \mathbb{E}_{p(x)\,\pi_0(\bm{a} \mid x)}
\Big[\sum_{k} \alpha_k\, \hat w_k\, \Delta_k\Big] \\
&= V(\pi)
- \mathbb{E}_{p(x)\,\pi_0(\bm{a} \mid x)}
\Big[\sum_{k} \alpha_k\, w^{*}\, \Delta_k\Big] \\
&\quad + \mathbb{E}_{p(x)\,\pi_0(\bm{a} \mid x)}
\Big[\sum_{k} \alpha_k\, \hat w_k\, \Delta_k\Big]
\;\; {\scriptstyle(\because\ \text{Lemma~\ref{lem:transport}})} \\[8pt]
&\therefore\ \begin{aligned}[t]
\mathrm{Bias}\big(\hat V_{\mathrm{ED\text{-}DR}}\big)
&= \mathbb{E}_{p(\mathcal{D})}
\big[\hat V_{\mathrm{ED\text{-}DR}}\big] - V(\pi) \\
&= \mathbb{E}
\Big[\sum_{k} \alpha_k\,
\big(\hat w_k - w^{*}\big)\, \Delta_k\Big]
\end{aligned}
\end{align*}
\end{proof}

\medskip\noindent\textit{Proof of Corollary~\ref{cor:e-side}.}
\begin{align*}
&\mathrm{Bias}\big(\hat V_{\mathrm{ED\text{-}DR}}\big) \\
&= \mathbb{E}_{p(x)\,\pi_0(\bm{a} \mid x)}
\Big[\sum_{k} \alpha_k\,
\big(\hat w_k - w^{*}\big)\, \Delta_k\Big] \\
&= \mathbb{E}_{p(x)\,\pi_0(\bm{a} \mid x)}
\Big[\sum_{k} \alpha_k\, w_k^{\mathrm{IIPS}}\,
\tfrac{\bar e_k^{\pi}}{e_k}\, e_k\, (r - \hat r)\Big] \\
&\quad - \mathbb{E}_{p(x)\,\pi(\bm{a} \mid x)}
\Big[\sum_{k} \alpha_k\, e_k\, (r - \hat r)\Big]
\;\; {\scriptstyle(\because\ \hat e_k = e_k,\ \text{Lemma~\ref{lem:transport}})} \\
&= \mathbb{E}_{p(x)\,\pi(\bm{a} \mid x)}
\Big[\sum_{k} \alpha_k\, \bar e_k^{\pi}\, (r - \hat r)\Big] \\
&\quad - \mathbb{E}_{p(x)\,\pi(\bm{a} \mid x)}
\Big[\sum_{k} \alpha_k\, e_k\, (r - \hat r)\Big]
\;\; {\scriptstyle(\because\ \text{Lemma~\ref{lem:transport}})} \\
&= \mathbb{E}_{p(x)\,\pi(\bm{a} \mid x)}
\Big[\sum_{k} \alpha_k\, e_k\, (r - \hat r)\Big] \\
&\quad - \mathbb{E}_{p(x)\,\pi(\bm{a} \mid x)}
\Big[\sum_{k} \alpha_k\, e_k\, (r - \hat r)\Big]
= 0
\;\; {\scriptstyle(\because\ \text{Lemma~\ref{lem:tower}})}
\end{align*}
\qed

\medskip\noindent\textit{Proof of Corollary~\ref{cor:r-side}.}
\begin{align*}
&\mathrm{Bias}\big(\hat V_{\mathrm{ED\text{-}DR}}\big) \\
&= \mathbb{E}_{p(x)\,\pi_0(\bm{a} \mid x)}
\Big[\sum_{k} \alpha_k\,
\big(\hat w_k - w^{*}\big)\, \Delta_k\Big] \\
&= \mathbb{E}_{p(x)\,\pi_0(\bm{a} \mid x)}
\Big[\sum_{k} \alpha_k\,
\big(w_k^{\mathrm{IIPS}} - w^{*}\big)\, \Delta_k\Big]
\;\; {\scriptstyle(\because\ \hat e_k = \hat\theta_k(x))} \\
&= \mathbb{E}_{p(x)\,\pi(\bm{a} \mid x)}
\Big[\sum_{k} \alpha_k\, \Delta_k\Big] \\
&\quad - \mathbb{E}_{p(x)\,\pi(\bm{a} \mid x)}
\Big[\sum_{k} \alpha_k\, \Delta_k\Big]
= 0
\;\; {\scriptstyle(\because\ \text{Lemma~\ref{lem:transport}})}
\end{align*}
\qed

\subsection{Proofs of Propositions~\ref{prop:variance} and~\ref{prop:var}}

\medskip\noindent\textit{Proof of Proposition~\ref{prop:variance}.}
Let $\hat g(x) := \mathbb{E}_{\pi(\bm{a}' \mid x)}
\big[\sum_{k} \alpha_k\, \hat q_k(x, \bm{a}')\big]$ denote the DM term in
Eq.~\eqref{eq:eddr}.
\begin{align*}
&\mathbb{V}_{p(\mathcal{D})}\big(\hat V_{\mathrm{ED\text{-}DR}}\big) \\
&= \frac{1}{n}\,\mathbb{V}_{p(\mathcal{D})}\Big(\hat g
+ \sum_{k} \alpha_k\, \hat w_k\, \big(Y_k - \hat q_k\big)\Big)
\;\; {\scriptstyle(\because\ \text{i.i.d.})} \\
&= \frac{1}{n}\Big\{
\mathbb{E}_{p(x)\,\pi_0(\bm{a} \mid x)}
\Big[\mathbb{V}_{p(\bm{Y} \mid x, \bm{a})}\Big(
\sum_{k} \alpha_k\, \hat w_k\, Y_k\Big)\Big] \\
&\qquad + \mathbb{V}_{p(x)\,\pi_0(\bm{a} \mid x)}
\Big(\hat g + \sum_{k} \alpha_k\, \hat w_k\, \Delta_k\Big)\Big\} \\
&= \frac{1}{n}\Big\{
\mathbb{E}_{p(x)\,\pi_0(\bm{a} \mid x)}
\Big[\mathbb{V}_{p(\bm{Y} \mid x, \bm{a})}\Big(
\sum_{k} \alpha_k\, \hat w_k\, Y_k\Big)\Big] \\
&\qquad + \mathbb{E}_{p(x)}
\Big[\mathbb{V}_{\pi_0(\bm{a} \mid x)}
\Big(\sum_{k} \alpha_k\, \hat w_k\, \Delta_k\Big)\Big] \\
&\qquad + \mathbb{V}_{p(x)}
\Big(\hat g + \mathbb{E}_{\pi_0(\bm{a} \mid x)}
\Big[\sum_{k} \alpha_k\, \hat w_k\, \Delta_k\Big]\Big)\Big\} \\
&= \frac{1}{n}\Big\{
\mathbb{E}_{p(x)\,\pi_0(\bm{a} \mid x)}
\Big[\mathbb{V}_{p(\bm{Y} \mid x, \bm{a})}\Big(
\sum_{k} \alpha_k\, \hat w_k\, Y_k\Big)\Big] \\
&\qquad + \mathbb{E}_{p(x)}
\Big[\mathbb{V}_{\pi_0(\bm{a} \mid x)}
\Big(\sum_{k} \alpha_k\, \hat w_k\, \Delta_k\Big)\Big] \\
&\qquad + \mathbb{V}_{p(x)}
\Big(\mathbb{E}_{\pi_0(\bm{a} \mid x)}
\Big[\sum_{k} \alpha_k\, \hat w_k\, q_k\Big]\Big)\Big\} \\
&\qquad\quad
\rlap{$\scriptstyle(\because\
\text{(S1), Lemma~\ref{lem:transport}, Lemma~\ref{lem:tower}})$}
\end{align*}
Eq.~\eqref{eq:var-le} follows analogously.
\qed

\medskip\noindent\textit{Proof of Proposition~\ref{prop:var}.}
\begin{align*}
&\mathbb{V}_{p(\mathcal{D})}\big(\hat V_{\mathrm{ED\text{-}DR}}\big)
- \mathbb{V}_{p(\mathcal{D})}\big(\hat V_{\mathrm{IIPS}}\big) \\
&= \frac{1}{n}\,\mathbb{V}_{p(\mathcal{D})}\Big(\hat g
+ \sum_{k} \alpha_k\, w_k^{\mathrm{IIPS}}\, \big(Y_k - \hat q_k\big)\Big) \\
&\quad - \frac{1}{n}\,\mathbb{V}_{p(\mathcal{D})}\Big(
\sum_{k} \alpha_k\, w_k^{\mathrm{IIPS}}\, Y_k\Big)
\;\; {\scriptstyle(\because\ \text{i.i.d., (E1), (E2)})} \\
&= \frac{1}{n}\Big\{\mathbb{E}_{p(x)\,\pi_0(\bm{a} \mid x)}
\Big[\mathbb{V}_{p(\bm{Y} \mid x, \bm{a})}\Big(
\sum_{k} \alpha_k\, w_k^{\mathrm{IIPS}}\, Y_k\Big)\Big] \\
&\qquad + \mathbb{V}_{p(x)\,\pi_0(\bm{a} \mid x)}
\big(\hat g + S - \hat S\big)\Big\} \\
&\quad - \frac{1}{n}\Big\{\mathbb{E}_{p(x)\,\pi_0(\bm{a} \mid x)}
\Big[\mathbb{V}_{p(\bm{Y} \mid x, \bm{a})}\Big(
\sum_{k} \alpha_k\, w_k^{\mathrm{IIPS}}\, Y_k\Big)\Big] \\
&\qquad + \mathbb{V}_{p(x)\,\pi_0(\bm{a} \mid x)}(S)\Big\} \\
&= \frac{1}{n}\,\mathbb{V}_{p(x)\,\pi_0(\bm{a} \mid x)}
\big(\hat g + S - \hat S\big)
- \frac{1}{n}\,\mathbb{V}_{p(x)\,\pi_0(\bm{a} \mid x)}(S) \\
&= \frac{1}{n}\,\mathbb{V}_{p(x)\,\pi_0(\bm{a} \mid x)}
\big(S - \big(\hat S - \mathbb{E}_{\pi_0(\bm{a} \mid x)}[\hat S]\big)\big) \\
&\quad - \frac{1}{n}\,\mathbb{V}_{p(x)\,\pi_0(\bm{a} \mid x)}(S)
\;\; {\scriptstyle(\because\ \text{(E1), (E2), Lemma~\ref{lem:transport}})} \\
&= \frac{1}{n}\,\mathbb{V}_{p(x)\,\pi_0(\bm{a} \mid x)}(S - \hat D)
- \frac{1}{n}\,\mathbb{V}_{p(x)\,\pi_0(\bm{a} \mid x)}(S) \\
&= \frac{1}{n}\Big\{\mathbb{E}_{p(x)}
\big[\mathbb{V}_{\pi_0(\bm{a} \mid x)}(S - \hat D)\big] \\
&\qquad + \mathbb{V}_{p(x)}
\big(\mathbb{E}_{\pi_0(\bm{a} \mid x)}[S - \hat D]\big)\Big\} \\
&\quad - \frac{1}{n}\Big\{\mathbb{E}_{p(x)}
\big[\mathbb{V}_{\pi_0(\bm{a} \mid x)}(S)\big]
+ \mathbb{V}_{p(x)}
\big(\mathbb{E}_{\pi_0(\bm{a} \mid x)}[S]\big)\Big\} \\
&= \frac{1}{n}\,\mathbb{E}_{p(x)}
\big[\mathbb{V}_{\pi_0(\bm{a} \mid x)}(S - \hat D)\big] \\
&\quad - \frac{1}{n}\,\mathbb{E}_{p(x)}
\big[\mathbb{V}_{\pi_0(\bm{a} \mid x)}(S)\big]
\;\; {\scriptstyle(\because\ \mathbb{E}_{\pi_0(\bm{a} \mid x)}[\hat D] = 0)} \\
&= \frac{1}{n}\,\mathbb{E}_{p(x)}
\big[\mathbb{V}_{\pi_0(\bm{a} \mid x)}(D - \hat D)\big]
- \frac{1}{n}\,\mathbb{E}_{p(x)}
\big[\mathbb{V}_{\pi_0(\bm{a} \mid x)}(D)\big] \\
&= \frac{1}{n}\,\mathbb{V}_{p(x)\,\pi_0(\bm{a} \mid x)}(D - \hat D)
- \frac{1}{n}\,\mathbb{V}_{p(x)\,\pi_0(\bm{a} \mid x)}(D)
\end{align*}
\qed

\subsection{Derivation of Eq.~\eqref{eq:iips-bias} and Proof of Proposition~\ref{prop:crossover}}

The bias of IIPS in Eq.~\eqref{eq:iips-bias} is derived as follows.
\begin{align*}
&b_{\mathrm{IIPS}}
:= \mathbb{E}_{p(\mathcal{D})}\big[\hat V_{\mathrm{IIPS}}\big] - V(\pi) \\
&= \mathbb{E}_{p(x)\,\pi_0(\bm{a} \mid x)}
\Big[\sum_{k=1}^{K} \alpha_k\, w_k^{\mathrm{IIPS}}\, q_k\Big] \\
&\quad
- \mathbb{E}_{p(x)\,\pi(\bm{a} \mid x)}
\Big[\sum_{k=1}^{K} \alpha_k\, q_k\Big]
\;\;\rlap{$\scriptstyle(\because\ \text{i.i.d.})$} \\
&= \mathbb{E}_{p(x)}\Big[\sum_{k=1}^{K} \alpha_k
\sum_{a \in \mathcal{A}} \pi(a \mid x, k) \\
&\quad
\big(\mathbb{E}_{\pi_0(\bm{a} \mid x)}[q_k \mid \bm{a}(k) = a]
- \mathbb{E}_{\pi(\bm{a} \mid x)}[q_k \mid \bm{a}(k) = a]\big)\Big] \\
&= \mathbb{E}_{p(x)}\Big[\sum_{k=1}^{K} \alpha_k
\sum_{a \in \mathcal{A}} \pi(a \mid x, k) \\
&\quad
\big(\bar e_k^{\pi_0}(x, a) - \bar e_k^{\pi}(x, a)\big)\, r(x, a)\Big]
\;\;\rlap{$\scriptstyle(\because\ \text{(A3)})$} \\
&= \mathbb{E}_{p(x)\,\pi(\bm{a} \mid x)}
\Big[\sum_{k=1}^{K} \alpha_k \\
&\quad
\big(\bar e_k^{\pi_0}(x, \bm{a}(k)) - \bar e_k^{\pi}(x, \bm{a}(k))\big)\, r(x, \bm{a}(k))\Big]
\end{align*}

\medskip\noindent\textit{Proof of Proposition~\ref{prop:crossover}.}
\begin{align*}
\mathrm{MSE}_{\mathrm{IIPS}} - \mathrm{MSE}_{\mathrm{ED\text{-}DR}}
&= b_{\mathrm{IIPS}}^2 + \frac{\sigma_{\mathrm{IIPS}}^2}{n}
- \frac{\sigma_{\mathrm{EDDR}}^2}{n}
> 0 \\
&\iff
n > \frac{\sigma_{\mathrm{EDDR}}^2 - \sigma_{\mathrm{IIPS}}^2}{b_{\mathrm{IIPS}}^2}
\end{align*}
\qed


\begin{thebibliography}{10}
\providecommand{\url}[1]{\texttt{#1}}
\providecommand{\urlprefix}{URL }
\providecommand{\doi}[1]{https://doi.org/#1}

\bibitem{ih2019}
Agarwal, A., Zaitsev, I., Wang, X., Li, C., Najork, M., Joachims, T.:
  Estimating position bias without intrusive interventions. In: Proc. 12th ACM
  Int. Conf. Web Search and Data Mining (WSDM). pp. 474--482. Melbourne, VIC,
  Australia (Feb 2019)

\bibitem{beygelzimer2009}
Beygelzimer, A., Langford, J.: The offset tree for learning with partial
  labels. In: Proc. 15th ACM SIGKDD Int. Conf. Knowledge Discovery and Data
  Mining (KDD). pp. 129--138. Paris, France (Jun 2009)

\bibitem{dbn}
Chapelle, O., Zhang, Y.: A dynamic {Bayesian} network click model for web
  search ranking. In: Proc. 18th Int. Conf. World Wide Web (WWW). pp. 1--10.
  Madrid, Spain (Apr 2009)

\bibitem{debiassurvey}
Chen, J., Dong, H., Wang, X., Feng, F., Wang, M., He, X.: Bias and debias in
  recommender system: A survey and future directions. ACM Trans. Information
  Systems  \textbf{41}(3),  67:1--67:39 (2023)

\bibitem{dml}
Chernozhukov, V., Chetverikov, D., Demirer, M., Duflo, E., Hansen, C., Newey,
  W., Robins, J.: Double/debiased machine learning for treatment and structural
  parameters. The Econometrics Journal  \textbf{21}(1),  C1--C68 (2018)

\bibitem{clickmodels}
Chuklin, A., Markov, I., de~Rijke, M.: Click Models for Web Search. Synthesis
  Lectures on Information Concepts, Retrieval, and Services, Morgan {\&}
  Claypool Publishers (2015)

\bibitem{craswell2008}
Craswell, N., Zoeter, O., Taylor, M.J., Ramsey, B.: An experimental comparison
  of click position-bias models. In: Proc. 1st ACM Int. Conf. Web Search and
  Data Mining (WSDM). pp. 87--94. Palo Alto, CA, USA (Feb 2008)

\bibitem{dudik2014}
Dud{\'{\i}}k, M., Erhan, D., Langford, J., Li, L.: Doubly robust policy
  evaluation and optimization. Statistical Science  \textbf{29}(4),  485--511
  (2014)

\bibitem{dudik2011}
Dud{\'{\i}}k, M., Langford, J., Li, L.: Doubly robust policy evaluation and
  learning. In: Proc. 28th Int. Conf. Machine Learning (ICML). pp. 1097--1104.
  Bellevue, WA, USA (Jun 2011)

\bibitem{twotower}
Hager, P., Zoeter, O., de~Rijke, M.: Unidentified and confounded? understanding
  two-tower models for unbiased learning to rank. In: Proc. ACM SIGIR Int.
  Conf. Innovative Concepts and Theories in Information Retrieval (ICTIR). pp.
  347--357. Padua, Italy (Jul 2025)

\bibitem{horvitz1952}
Horvitz, D.G., Thompson, D.J.: A generalization of sampling without replacement
  from a finite universe. Journal of the American Statistical Association
  \textbf{47}(260),  663--685 (1952)

\bibitem{ultr}
Joachims, T., Swaminathan, A., Schnabel, T.: Unbiased learning-to-rank with
  biased feedback. In: Proc. 10th ACM Int. Conf. Web Search and Data Mining
  (WSDM). pp. 781--789. Cambridge, UK (Feb 2017)

\bibitem{lips}
Kiyohara, H., Nomura, M., Saito, Y.: Off-policy evaluation of slate bandit
  policies via optimizing abstraction. In: Proc. ACM Web Conf. (WWW). pp.
  3150--3161. Singapore (May 2024)

\bibitem{cascadedr}
Kiyohara, H., Saito, Y., Matsuhiro, T., Narita, Y., Shimizu, N., Yamamoto, Y.:
  Doubly robust off-policy evaluation for ranking policies under the cascade
  behavior model. In: Proc. 15th ACM Int. Conf. Web Search and Data Mining
  (WSDM). pp. 487--497. Tempe, AZ, USA (Feb 2022)

\bibitem{aips}
Kiyohara, H., Uehara, M., Narita, Y., Shimizu, N., Yamamoto, Y., Saito, Y.:
  Off-policy evaluation of ranking policies under diverse user behavior. In:
  Proc. 29th ACM SIGKDD Int. Conf. Knowledge Discovery and Data Mining (KDD).
  pp. 1154--1163. Long Beach, CA, USA (Aug 2023)

\bibitem{iips}
Li, S., Abbasi-Yadkori, Y., Kveton, B., Muthukrishnan, S., Vinay, V., Wen, Z.:
  Offline evaluation of ranking policies with click models. In: Proc. 24th ACM
  SIGKDD Int. Conf. Knowledge Discovery and Data Mining (KDD). pp. 1685--1694.
  London, UK (Aug 2018)

\bibitem{rips}
McInerney, J., Brost, B., Chandar, P., Mehrotra, R., Carterette, B.A.:
  Counterfactual evaluation of slate recommendations with sequential reward
  interactions. In: Proc. 26th ACM SIGKDD Int. Conf. Knowledge Discovery and
  Data Mining (KDD). pp. 1779--1788. Virtual Event, CA, USA (Aug 2020)

\bibitem{drultr}
Oosterhuis, H.: Doubly robust estimation for correcting position bias in click
  feedback for unbiased learning to rank. ACM Trans. Information Systems
  \textbf{41}(3),  61:1--61:33 (2023)

\bibitem{precup2000}
Precup, D., Sutton, R.S., Singh, S.: Eligibility traces for off-policy policy
  evaluation. In: Proc. 17th Int. Conf. Machine Learning (ICML). pp. 759--766.
  Stanford, CA, USA (Jun 2000)

\bibitem{richardson2007}
Richardson, M., Dominowska, E., Ragno, R.: Predicting clicks: Estimating the
  click-through rate for new ads. In: Proc. 16th Int. Conf. World Wide Web
  (WWW). pp. 521--530. Banff, AB, Canada (May 2007)

\bibitem{obp}
Saito, Y., Aihara, S., Matsutani, M., Narita, Y.: Open bandit dataset and
  pipeline: Towards realistic and reproducible off-policy evaluation. In: Proc.
  35th Conf. Neural Information Processing Systems (NeurIPS), Track on Datasets
  and Benchmarks. Virtual Event (Dec 2021)

\bibitem{mnar}
Saito, Y., Yaginuma, S., Nishino, Y., Sakata, H., Nakata, K.: Unbiased
  recommender learning from missing-not-at-random implicit feedback. In: Proc.
  13th ACM Int. Conf. Web Search and Data Mining (WSDM). pp. 501--509. Houston,
  TX, USA (Feb 2020)

\bibitem{rat}
Schnabel, T., Swaminathan, A., Singh, A., Chandak, N., Joachims, T.:
  Recommendations as treatments: Debiasing learning and evaluation. In: Proc.
  33rd Int. Conf. Machine Learning (ICML). vol. PMLR 48, pp. 1670--1679. New
  York, NY, USA (Jun 2016)

\bibitem{strehl2010}
Strehl, A.L., Langford, J., Li, L., Kakade, S.M.: Learning from logged implicit
  exploration data. In: Proc. 24th Conf. Neural Information Processing Systems
  (NeurIPS). pp. 2217--2225. Vancouver, BC, Canada (Dec 2010)

\bibitem{pi}
Swaminathan, A., Krishnamurthy, A., Agarwal, A., Dud{\'{\i}}k, M., Langford,
  J., Jose, D., Zitouni, I.: Off-policy evaluation for slate recommendation.
  In: Proc. 31st Conf. Neural Information Processing Systems (NeurIPS). pp.
  3632--3642. Long Beach, CA, USA (Dec 2017)

\bibitem{vlassis}
Vlassis, N., Chandrashekar, A., Gil, F.A., Kallus, N.: Control variates for
  slate off-policy evaluation. In: Proc. 35th Conf. Neural Information
  Processing Systems (NeurIPS). pp. 3667--3679. Virtual Event (Dec 2021)

\bibitem{regem}
Wang, X., Golbandi, N., Bendersky, M., Metzler, D., Najork, M.: Position bias
  estimation for unbiased learning to rank in personal search. In: Proc. 11th
  ACM Int. Conf. Web Search and Data Mining (WSDM). pp. 610--618. Marina Del
  Rey, CA, USA (Feb 2018)

\end{thebibliography}
\end{document}